\documentclass{iopart}

\usepackage{bm}
\usepackage{graphics}      

\expandafter\let\csname equation*\endcsname\relax 
\expandafter\let\csname endequation*\endcsname\relax 
\usepackage{amsmath}
\usepackage{epsfig}
\usepackage{graphics}
\usepackage{hyperref}
\usepackage{amsthm}
\usepackage{iopams}
\usepackage{braket}
\usepackage{bbold}					
\usepackage{color}

\newtheorem*{theorem}{Theorem}
\newtheorem*{lemma}{Lemma}
\newtheorem*{corollary}{Corollary}

\begin{document}

\date{\today}
\title{General conditions for maximal violation of non-contextuality in discrete and continuous variables}
\author{A.~Laversanne-Finot$^1$, A.~Ketterer$^1$, M. R. Barros$^2$, S. P. Walborn$^2$, T.~Coudreau$^1$, A.~Keller$^{1, 3}$, P.~Milman$^1$}
\eads{\mailto{adrien.laversanne-finot@univ-paris-diderot.fr},
\mailto{perola.milman@univ-paris-diderot.fr}}
\address{$^1$ Univ Paris Diderot, Sorbonne Paris Cit\'e, MPQ, UMR 7162 CNRS, F-75205 Paris, France}
\address{$^2$ Instituto de F\'{\i}sica, Universidade Federal do Rio de
Janeiro, Caixa Postal 68528, Rio de Janeiro, RJ 21941-972, Brazil}
\address{$^3$ Universit\'e Paris-Sud, Universit\'e Paris-Saclay 91405 Orsay, France}

\begin{abstract}
The contextuality of quantum mechanics can be shown by the violation of inequalities based on measurements of well chosen observables. An important property of such observables is that their expectation value can be expressed in terms of probabilities for obtaining two exclusive outcomes. Examples of such inequalities have been constructed using either observables with a dichotomic spectrum or using periodic functions obtained from displacement operators in phase space. Here we identify  the general conditions on the spectral decomposition of observables demonstrating state independent contextuality of quantum mechanics. Our results not only unify existing strategies for maximal violation of state independent non-contextuality inequalities but also lead to new scenarii enabling such violations. Among the consequences of our results is the impossibility of having a state independent maximal violation of non-contextuality in the Peres-Mermin scenario with discrete observables of odd dimensions.
\end{abstract}
%


%
\section{Introduction}
%
The question of whether physical systems have intrinsic properties is a long standing debate~\cite{PhysRev.47.777} that was turned upside down with the advent of quantum physics. A core result by Bell states that there is no local hidden variable model reproducing the predictions of quantum mechanics~\cite{Bell:1964}. Another important result is the Kochen-Specker theorem \cite{Kochen:1967aa} that demonstrates the contextual nature of quantum mechanics.

In a non-contextual theory, the result of a measurement $\nu(A)$ depends only on the state of the system and the observable ${A}$ being measured. Additionally, measurement outcomes can depend on some (possibly hidden) variable $\lambda$ describing the state of the system. If one knows $\lambda$ then one can predict the outcome of any measurement: we say thus that measurement outcomes are pre-determined. This corresponds to the classical view in which every system is in a well defined state. In particular, in a non-contextual theory measurement outcomes do not depend on the compatible observables that are measured together with ${A}$. 

The initial argument by Kochen-Specker to show the contextuality of quantum mechanics used a set of 117 vectors in a 3-dimensional space \cite{Kochen:1967aa}. Since then, many attempts have been made to refine this argument and turn it into an experimentally testable property. The contextuality of quantum mechanics was proved for a particular state and a Hilbert space of dimension 4 by Peres \cite{Peres:PhysLettA:151}. Mermin showed that this argument could be recast to find a state independent proof of contextuality \cite{PhysRevLett.65.3373}. The same type of argument as Mermin's was used to derive state independent non-contextuality inequalities ($\mathit{i.e.}$, that can be violated by any state if non-contextuality does not hold)~\cite{PhysRevLett.101.210401, PhysRevLett.109.250402, PhysRevLett.108.030402}. Such inequalities, obtained in the so--called Peres Mermin scenario (PMS) are particularly attractive from the experimental point of view, and  have been experimentally tested with trapped ions~\cite{Kirchmair:2009aa}, nuclear spin ensembles~\cite{PhysRevLett.104.160501} and photons~\cite{PhysRevX.3.011012, PhysRevLett.108.200405, PhysRevLett.103.160405}. In addition, it has been proven that contextuality (in a state dependent formulation) is a critical resource for quantum computing~\cite{PhysRevA.88.022322, NJPProba, Howard:2014aa}. 

Even though the study of contextuality was originally focused on discrete variable system, such as qubits and qudits, it is also possible to find state independent non-contextuality inequalities for continuous variables in the PMS~\cite{PhysRevA.82.022114, Asadian:2015aa}. In this case, one notes that the operators used to derive the inequalities have a bounded spectrum. This last property ensures that their expectation values can be expressed as the ones of dichotomic observables defined in an extended space~\cite{PhysRevA.67.060101}. The bounded observables used in~\cite{PhysRevA.82.022114, Asadian:2015aa} can be obtained by measuring bounded functions of observables with an arbitrary spectrum, as considered in the protocols described in~\cite{Asadian:2015aa} and~\cite{PhysRevA.82.022114}. A similar technique led, in~\cite{PhysRevA.91.012106, PhysRevA.92.062107}, to the definition of dimension independent Bell-type inequalities~\cite{PhysRev.47.777, Bell:1964, PhysRevLett.23.880}, which is a particular case of non-contextuality inequalities where, in addition, locality is enforced. Ruling out local realism in experiments requires to satisfy more stringent constraints that are not necessary to prove the contextuality of quantum mechanics per se. The contextuality of quantum mechanics can be proven, in principle, by measuring well chosen observables, independently of the system's particular state~\cite{PhysRevLett.103.050401}. It is thus of interest to characterize which general properties observables must have for testing contextuality and to maximally violate experimentally testable non-contextuality inequalities.  

So far, the contextuality of quantum mechanics has been shown for specific observables defined by continuous or discrete variables. In addition, according to the considered case, the border between contextual and non-contextual theories varies.  It is natural to seek to identify the common features of the existing results and try to formalize the general conditions quantum observables must fulfill in order to demonstrate state independent contextuality irrespectively of their dimensionality. Such understanding would potentially enable the state independent test of this essential property of quantum mechanics in any quantum system. In other words, what are the common/distinctive properties and features of non-contextuality inequalities? How can one build a suitable inequality from arbitrary observables permitting the demonstration of state independent contextuality in quantum mechanics? 

In this article we answer these questions in the PMS approach, which is, as mentioned, a particularly experimentally attractive formulation of the Kochen-Specker theorem. The PMS is shown in terms of Table~\ref{square1}, and was originally constructed to illustrate the differences between contextual and non-contextual theories in the case of measurements of dichotomic observables that, in quantum mechanics, can be represented by the Pauli matrices. We show that the generalized version of the PMS can be obtained using complex functions (continuous or discrete) of modulus one instead of real ones. In quantum mechanics these functions are associated to unitary operations acting on a bipartite system. We will see that state independent maximal violation of non-contextuality inequalities is possible iff, for each one of the considered parties, the introduced unitary operators obey specific commutation relations. Theses relations reduce to the known particular cases according to the chosen set-up. A consequence of our results is that in the bi-partite PMS with discrete variables, it is only possible to observe state independent maximal violation of the non-contextual bound with qudits of even dimension. Furthermore, we derive necessary and sufficient conditions the spectrum of the observable must fulfill such that it can be used for state independent maximal violation of the non-contextual bound. Our results significantly expand the possibilities of experimental state independent maximal violation of non contextual inequalities in arbitrary dimensions  and provide a clear and unified framework for demonstrating the contextuality of quantum mechanics.

This paper is organized as follows: in section~\ref{sec:PMS} we start by recalling the principles of the PMS before generalizing it in~\ref{sec:PMSarbitrary}. In this section we prove the main result of this paper, characterizing the spectrum of observables that can be
used to perform a contextuality test. We then discuss some examples in section~\ref{sec:Examples} and show how our result apply to previously known observables and can be used to find new observables suitable for contextuality tests. Finally, we conclude in~\ref{sec:Conclusion}.
\section{The Peres-Mermin square\label{sec:PMS}}
Let us consider a set of nine dichotomic observables $\{A_{jk}\}$, $i,j=1,2,3$, such that the observables sharing a common subscript are mutually commuting. From this set of observables, one can construct the quantity 
\begin{align}\label{contextualInequality}
\braket{X} = & \braket{A_{11}A_{12}A_{13}}+\braket{A_{21}A_{22}A_{23}}+\braket{A_{31}A_{32}A_{33}} \\
&+\braket{A_{11}A_{21}A_{31}}+\braket{A_{12}A_{22}A_{32}}-\braket{A_{13}A_{23}A_{33}}. \nonumber
\end{align}
In a non-contextual theory, observables are described by pre-determined values $-1$ or $1$. One can show, by testing every possible combination of outcomes for the $\{A_{jk}\}$, that the maximum value of $\braket{X}$ is 4 in  a non-contextual theory~\cite{PhysRevLett.101.210401}.

\begin{table}[!h]
\caption{\label{square1}The Peres-Mermin square for Pauli operators,  $\hat{\sigma}_{x, y, z}$.}
\begin{center}
\begin{tabular}{cccc}
  $A_{jk}$ & $k=1$ & $k=2$ & $k=3$ \\ 
  \hline
  $j=1$ & $\hat{\sigma}_x \otimes\mathbb{1}$ & $\mathbb{1}\otimes \hat{\sigma}_x$ & $\hat{\sigma}_x\otimes \hat{\sigma}_x$\\
  $j=2$ & $\mathbb{1}\otimes \hat{\sigma}_z$ & $\hat{\sigma}_z\otimes\mathbb{1}$ & $\hat{\sigma}_z\otimes \hat{\sigma}_z$ \\
  $j=3$ & $\hat{\sigma}_x\otimes \hat{\sigma}_z$ & $\hat{\sigma}_z\otimes \hat{\sigma}_x$ & $\hat{\sigma}_y\otimes \hat{\sigma}_y$
\end{tabular}
\end{center}
\end{table}
We now consider the case where observables $\{A_{ij}\}$ are quantum and given by Table~\ref{square1}. One can easily check that the observables in the same row or column are compatible (commuting).  Nevertheless, because the product of operators along each row and column is $\mathbb{1}$, except for the last column where it is $-\mathbb{1}$, one has that for every quantum state  $\braket{X}_{QM} = 6$, violating the bound of 4 discussed above and thus proving that quantum mechanics is contextual in the case where dichotomic observables are measured. 

The question we address now is how to generalize the PMS in order to extend the tests of the contextuality of quantum mechanics to situations where  observables with fundamentally different properties from the Pauli matrices are measured. Our generalization will permit the identification of conditions that observables with an arbitrary spectrum must satisfy in order to prove contextuality.
\section{Peres-Mermin square with arbitrary unitary operators\label{sec:PMSarbitrary}}
%
%
Contextuality can also be tested using complex functions instead of real ones, as is the case of the previous example involving Pauli matrices. This leads to inequalities involving the (independent) real and imaginary parts of such functions. Of course, in order to test contextuality, experiments involving either the measurement of the real or the imaginary part of such functions must be carried out. Also, it is clear that since the real and imaginary parts are independent, contextuality can be independently tested for each one of these components. 

Enlightening results that will be used here as a guideline were obtained by Asadian {\it et al.}~\cite{Asadian:2015aa}, where the particular case of contextuality tests using displacements in phase space was studied. There, the authors obtain many interesting conditions and constraints for testing contextuality using displacement operators that can be well understood in the light of the general framework we obtain here. 

Non-contextuality inequalities involving complex functions can be derived by choosing  the $A_{jk}$'s appearing in the PMS to be complex functions $U_{jk}= A^R_{jk} + iA^I_{jk}$, with $|A^R_{jk}|^2+|A^I_{jk}|^2=1$. This leads to  Table~\ref{square2}
\begin{table}[!h]
\caption{\label{square2}The Peres-Mermin square for arbitrary operators, where $U_i$ are arbitrary unitary operators.}
\begin{center}
\begin{tabular}{cccc}
  $A_{jk}$ & $k=1$ & $k=2$ & $k=3$ \\ 
  \hline
  $j=1$ & $\hat{U}_1^{\dagger}\otimes\mathbb{1}$ & $\mathbb{1}\otimes \hat{U}_1^{\dagger}$ & $\hat{U}_1\otimes \hat{U}_1$\\
  $j=2$ & $\mathbb{1}\otimes \hat{U}_2^{\dagger}$ & $\hat{U}_2^{\dagger}\otimes\mathbb{1}$ & $\hat{U}_2\otimes \hat{U}_2$ \\
  $j=3$ & $\hat{U}_1\otimes \hat{U}_2$ & $\hat{U}_2\otimes \hat{U}_1$ & $\hat{U}_3\otimes \hat{U}_3$ \\
\end{tabular}
\end{center}
\end{table}
where, in quantum mechanics, functions  $U_{j}$ become unitary operators, $\hat U_{j}$. For the sake of clarity, we used here a similar reasoning and notation as the one in~\cite{Asadian:2015aa}, with the important difference that while Ref.~\cite{Asadian:2015aa} was  restricted to the specific case of displacement operators, here we consider that operators $\hat{U}_j$ can be arbitrary unitaries defined in a Hilbert space of arbitrary dimension. By doing so, we can identify Table~\ref{square1} as a particular case of Table~\ref{square2}. 

By  multiplying the PMS's rows and columns in  Table~\ref{square2} , we are left with an inequality involving complex functions. It can be transformed in an inequality for real functions by taking its real or imaginary parts. We will consider here its real part:
\begin{equation}\label{realpartcontextuality}
\braket{\mathrm{Re}(X)} = \braket{R_1} + \braket{R_2} + \braket{R_3} + \braket{C_1} + \braket{C_2} - \braket{C_3} ,
\end{equation}
where
\begin{align}\label{expansion1}
R_j = (A^R_{j1}A^R_{j2}-A^I_{j1}A^I_{j2})A^R_{j3}-(A^I_{j1}A^R_{j2}+A^R_{j1}A^I_{j2})A^I_{j3}, \\
C_k = (A^R_{1k}A^R_{2k}-A^I_{1k}A^I_{2k})A^R_{3k}-(A^I_{1k}A^R_{2k}+A^R_{1k}A^I_{2k})A^I_{3k}.\label{expansion2}
\end{align}
In~\cite{PhysRevA.82.022114} it was proven that for non-contextual theories where the condition $\braket{(A_{ij}^R)^2 + (A_{ij}^I)^2} \le 1$ is assumed, $\braket{\mathrm{Re}(X)} \leq 3\sqrt{3}$. In fact, it is not necessary to assume this condition. Indeed, it has been proving in~\cite{Asadian:2015aa} that one can also probe contextuality for continuous variables without any assumptions, albeit with a slightly modified inequality. The description made so far is general and encompass the PMS with Pauli operators. Indeed, if one adds the additional constraint that the real (or imaginary) part of the considered complex functions, $U_{j}$, is zero, together with the condition that the measured observables must have a dichotomic expectation value, one recovers Table~\ref{square1} and, from (\ref{realpartcontextuality}), the non-contextual bound of $4$.

We now move to the quantum description of the PMS using unitary operators. Since unitary operators are not, in general, observables, one can split them into their real and imaginary Hermitian parts, $\hat U_{jk} = \hat A^R_{jk} + i\hat A^I_{jk}$, which are observables. In the following lemma we show that these observables will maximally violate the non-contextuality inequality $\braket{\mathrm{Re}(\hat{X})} \le 3$ if the unitary operators obey specific commutation and anti-commutation relations.
\begin{lemma}[Maximal state-independent contextuality]
The operators $\hat{U}_1$, $\hat{U}_2$ and $\hat{U}_3$ will lead to a state independent maximal violation of the non-contextual bound if and only if they satisfy the following commutation and anti-commutation relations:
\begin{align}
 [\hat U_i, \hat U_j] &= \pm 2i\epsilon_{ijk}\hat U_k^{\dagger}, \label{commute2a}\\
 \{\hat U_i, \hat U_j\} &= 2\delta_{ij}\hat U_i^2, \label{commute2b}
\end{align}
where $\epsilon_{ijk}$ is the Levi-Civita symbol.
\end{lemma}
\begin{proof}
From the PMS in Table~\ref{square2}, we can see that, in order to maximally violate the non-contextual bound, the product of the three operators in each row and column must be $\mathbb{1}$ except in the last column where it must be $-\mathbb{1}$. Also, unitaries in the same row or column must be compatible, leading to the constraints on the commutator $[\hat U_1, \hat U_3] = 0$ or on the anti-commutator $\{\hat U_1, \hat U_3\} = 0$ and the same for $\hat{U}_2$ and $\hat{U}_3$. These conditions cannot be verified at the same time, and the only possibility to obtain a state independent maximal violation of the non-contextual bound is to enforce $\{\hat U_1, \hat U_3\} = 0$ and $\{\hat U_2, \hat U_3\} = 0$. All the above ingredients combined lead to the following conditions for maximal violation of non-contextuality inequalities based on the PMS: $\hat U_1\hat U_2\hat U_3 = \pm i\mathbb{1}$ and $\hat U_2 \hat U_1\hat U_3 = \mp i\mathbb{1}$, which are equivalent to:
\begin{align}
 \pm i \hat U_2^{\dagger}\hat U_1^{\dagger}&= \hat U_3 , \label{pauliRelations1}\\
 \{\hat U_1, \hat U_2\} &= 0.  \label{pauliRelations2}
\end{align}
From the conditions~(\ref{pauliRelations1}) and (\ref{pauliRelations2}) we see that for state independent maximal violation of the PM inequality the operators $U_1$ and $U_2$ must be anti-commuting and that they completely determine the operator $U_3$ that  completes the set. Thus, if the unitary operators in the PMS fulfill the commutation relations~(\ref{commute2a}) and (\ref{commute2b}) the expectation~(\ref{realpartcontextuality}) maximally violates the non-contextuality inequality with $\braket{\mathrm{Re}(X)}=6 $, for all states $\hat\rho$.
\end{proof}
The conditions~(\ref{commute2a}) and (\ref{commute2b}) are general, and to our knowledge, have not been established so far. Previous results showing the possibility of violation of the non-contextuality inequalities are particular cases obeying these conditions. Examples are state independent contextuality using two-level systems~\cite{PhysRevLett.101.210401} and displacement operators~\cite{Asadian:2015aa}. 

We now make a step further beyond the relations~(\ref{commute2a}) and (\ref{commute2b}), and answer to the following question: given a unitary operator $\hat U_1$, what are the necessary and sufficient conditions for finding two other operators $\hat U_2$ and $\hat U_3$ such that~(\ref{commute2a}) and (\ref{commute2b}) are satisfied and thus lead to a maximal violation of noncontextuality inequalities derived from the Peres-Mermin square? The answer to this question is addressed by the following result:
\begin{theorem}[Anti-commutation of unitary operators]\label{theorem:eigenspace}
A unitary operator $\hat U_1$, acting on a Hilbert space $\mathcal H$, admits an anti-commuting partner if and only if for each eigenvalue $\lambda$ of $\hat U_1$, we find a corresponding eigenvalue $-\lambda$ whose eigenspace has the same dimension $K$ as the one of $\lambda$.
\end{theorem}
\begin{proof}
To prove the above statement we assume first that $\hat{U}_1$ fulfills the above condition on the spectrum and prove that it admits an anti-commuting partner. We restrict ourselves here to a proof in the finite dimensional case. Let's define the set of eigenvalues of $\hat{U}_1$ as $\{\lambda_1, \ldots, \lambda_k, -\lambda_1,\ldots, -\lambda_k\}$, and the set of eigenvectors associated  to each of the eigenvalues $\pm \lambda_i$ as $\{\ket{e_{i,j}^{\pm}}\}$, with possible degeneracy $j \in \{1,\ldots, K_i\}$. Since $\hat{U}_1$ is a unitary operator, we know that the set of eigenvectors $\{\ket{e_{i,j}^{\pm}}\}$ represents an orthonormal basis of the Hilbert space. Further on, we define an operator $\hat{U}_2$ through: $\hat{U}_2\ket{e_{i,j}^{\pm}} =\lambda_i '\ket{e_{i,j}^{\mp}}$, where $\lambda_i'$ are arbitrary complex numbers with absolute value $1$, which maps an orthonormal basis to another orthonormal basis thus providing a unitary operator. A simple calculation yields:
\begin{align}
(\hat{U}_1\hat{U}_2 + \hat{U}_2\hat{U}_1)\ket{e_{i,j}^{\pm}} & =\lambda_i' \hat{U}_1\ket{e_{i,j}^{\mp}} \pm  \lambda_i \hat{U}_2\ket{e_{i,j}^{\pm}} \nonumber \\
& = \mp \lambda_i \lambda_i' \ket{e_{i,j}^{\mp}} \pm \lambda_i \lambda_i'  \ket{e_{i,j}^{\mp}} = 0. \label{eq:proofU2}
\end{align}
showing that $\hat{U}_1$ and $\hat{U}_2$ are anti-commuting.

To prove the converse statement let's assume that we have two unitary operators $\hat{U}_1$ and $\hat{U}_2$ satisfying $\{\hat{U}_1, \hat{U}_2\} = 0$. We denote by $\lambda$ an eigenvalue of $\hat{U}_1$ with the corresponding eigenvectors $\ket{\{e_{i}\}}$, where $i = 1\ldots K$. Using the anti-commutation relation we can prove that $\hat{U}_2\ket{e_{i}}$ is an eigenvector of $\hat{U}_1$ with eigenvalue $-\lambda$:
\begin{equation}
(\hat{U}_1\hat{U}_2 + \hat{U}_2\hat{U}_1) \ket{e_{i}} = \hat{U}_1\hat{U}_2 \ket{e_{i}} + \hat{U}_2\lambda \ket{e_{i}}.
\end{equation}
Hence:
\begin{equation}\label{eq:eigenvalue-1}
\{U_1, U_2\} = 0 \Rightarrow \hat{U}_1\hat{U}_2\ket{e_{i}} = -\lambda \hat{U}_2 \ket{e_{i}}.
\end{equation}
Since $\{\ket{e_{i}}\}$ is an orthonormal set and $\hat{U}_2$ is a unitary operator, $\{\hat{U}_2\ket{e_{i}}\}$ is also an orthonormal set, which proves that $-\lambda$ is an eigenvalue of $\hat{U}_1$ of dimension larger or equal than $K$. The same reasoning applied to the set of eigenvectors of $\hat{U}_1$ with eigenvalue $-\lambda$ to show that the dimension of the eigenspace associated to $\lambda$ is higher or equal than the dimension of the eigenspace associated to $-\lambda$ and thus equal.
\end{proof}
Note that, the above theorem must hold for all unitary operators in the PMS,$\hat{U}_1$, $\hat{U}_2$ and $\hat{U}_3$. As a consequence of the previous lemma and theorem, we find the following characterization of the operators contained in the PMS:
\begin{corollary}[Structure of anti-commuting operators]
Unitary operators $\hat{U}_1$, $\hat{U}_2$ and $\hat{U}_3$, which lead to a state-independent maximal violation of the Peres-Mermin inequality, can be expressed in some basis as:
\begin{align}
\hat{U}_1 & = \bigoplus_{i=1}^N \lambda_i \hat \sigma_z^{(i)} \label{eq:DirectSumU1} \\
\hat U_2 & = \bigoplus_{i=1}^N \lambda_i' \hat \sigma_x^{(i)},  \label{eq:DirectSumU2} \\
\hat U_3 & =\pm \bigoplus_{i=1}^N (\lambda_i \lambda_i')^* \hat \sigma_y^{(i)}, \label{eq:DirectSumU3}
\end{align}
where $\pm \lambda_i$ are the eigenvalues of $\hat U_1$, $\hat \sigma_z^{(i)}=\bigoplus_{j=1}^{K_i}\hat \sigma_z$ is a direct sum of Pauli operators acting on the eigenspace associated to the eigenvalue $\pm \lambda_i$ with degeneracy $K_i$, and $N$ is an arbitrary, possibly infinite, integer value that is smaller than the Hilbert space dimension. $\hat \sigma^{(i)}_x=\bigoplus_{j=1}^{K_i}\hat \sigma_x$ and $\hat \sigma^{(i)}_y=\bigoplus_{j=1}^{K_i}\hat \sigma_y$ are defined similarly.
\end{corollary}
\begin{proof}
An operator that fulfills the above theorem can be expressed in some basis as a direct sum:
\begin{equation}
\hat{U}_1 = \bigoplus_{i=1}^N \lambda_i \hat \sigma_z^{(i)}.
\end{equation}
where $\pm \lambda_i$ are the eigenvalues of $\hat U_1$. Let us denote by $\ket{e^{\pm}_{i, j}}$ the eigenvectors associated to $\pm \lambda_i$. As shown before, $\hat{U}_2 \ket{e^{\pm}_{i, j}}$ is an eigenvector of $\hat{U}_1$ with eigenvalue $\mp\lambda_i$ and so the only non zero elements of $\hat{U}_2$ are $\braket{e^{\pm}_{i, j} | \hat U_2 | e^{\mp}_{i, j'}}=\lambda_i^\prime$. From this it follows directly that $\hat U_2$ can be expressed in the following form:
\begin{align}
\hat U_2 = \bigoplus_{i=1}^N \lambda_i' \sigma_x^{(i)},
\end{align}
where $\hat \sigma^{(i)}_x=\bigoplus_{j=1}^{K_i}\hat \sigma_x$ is a direct sum of Pauli operators defined on the same two dimensional space as $\hat{\sigma}_z$. Finally, we can simply use Eq.~(\ref{pauliRelations2}) to calculate 
\begin{equation}
\hat U_3 = \pm \bigoplus_{i=1}^N (\lambda_i \lambda_i')^* \hat \sigma_y^{(i)}.
\end{equation}
\end{proof}
Note that, consequently, the diagonalization of a set of unitary operators, $\hat U_1$, $\hat U_2$ and $\hat{U}_3$, which leads to a state-independent violation of the Peres-Mermin inequality, will always yield the same binary form, as shown in Eq.~(\ref{eq:DirectSumU1}). This shows that maximal state-independent contextuality in the PMS is a very peculiar property related to spectrum of operators whose spectral decomposition, continuous or discrete, can be written in terms of finite or infinite direct sums of Pauli matrices weighted by complex numbers of modulus one.
\section{State independent violation of contextuality\label{sec:Examples}}
We will now study some examples of operators satisfying the presented conditions and show how they relate to the known Peres-Mermin scenario. In this respect, we will first focus on the finite dimensional case and show to demonstrate state-independent contextuality in terms of spin systems. Subsequently, we turn to the case of infinite dimensional Hilbert spaces for which we discuss two prominent examples of unitary operators that allow to rule out non-contextuality.

\subsection{Finite dimensional case}
The decompositions~(\ref{eq:DirectSumU1}) and (\ref{eq:DirectSumU2}) reveal the binary structure of the spectrum of the unitary operators $\hat U_i$, with $i=1,2,3$, which is at the heart of a maximal violation of the Peres-Mermin non-contextuality inequality for finite $N$. Thus, state independent maximal violation of contextuality in a Peres-Mermin scenario is only possible in a Hilbert space of even dimension and formed by two parties which are themselves also of even dimension. In~\cite{Asadian:2015aa}, the authors reached a similar conclusion for the case of discrete displacements in phase space. Thanks to the generality of the conditions obtained here, we can analyze in more detail a scenario containing measurements of finite discrete dimensional quantum systems, so-called qudits.

To begin let's consider the simplest case of qubit measurements, corresponding to $N=1$ in Eq.~(\ref{eq:DirectSumU1}) and (\ref{eq:DirectSumU2}), for which we recover the Peres-Mermin scenario discussed in Sec.~\ref{sec:PMS} with the Peres-Mermin squares depicted in Table~\ref{square1}. When moving to higher dimensional systems, for instance, a pair of spin $S$ particles, contextuality can be demonstrated using the following rotation operators: 
\begin{align}
\hat R_1 = e^{i\hat S_xt_1}, \  \hat R_2 = e^{i\hat S_yt_2}, \  \hat R_3 = e^{i\hat S_zt_3},
\label{eq:RotSpinS}
\end{align}
where $\hat S_x$, $\hat S_y$ and $\hat S_z$ are the three vector components of the spin $S$ operator $\hat{\textbf S}$, generating the group $\mathrm{SU}(2)$ of all unitary rotations in a $d=2S+1$ dimensional Hilbert space. In order to build a Peres-Mermin square, one must choose $t_1$, $t_2$ and $t_3$ such that $R_1$, $R_2$ and $R_3$ verify (\ref{eq:DirectSumU1}). The matrix elements of the $z$-component of $\hat{\textbf S}$ read $(S_z)_{ab} = (S+1-b)\delta_{a,b}$, and the eigenvalues of $R_1$ are $\exp(i(S+1-b)t_1)$, for $b = 1, \ldots, d-1$. Hence, conditions~(\ref{eq:DirectSumU1}) and (\ref{eq:DirectSumU2}) are only satisfied if $t_3 = \pi$ and, since $S_x$ and $S_y$ are unitarily equivalent to $S_z$,  if $t_1 =t_2= \pi$. In this case, $R_1$, $R_2$ and $R_3$ lead to a maximal violation of the Peres-Mermin inequality in terms of rotations of half-integer spins, generalizing the qubit case presented in Table~\ref{square1}. 


\subsection{Infinite dimensional case} 
In this Section we study observables which are defined in infinite dimensional Hilbert spaces. We start by the famous photon-number parity operator and later on turn to the case of modular variable measurements. The latter provides an example of a contextuality test involving measurements of observables with continuous outcomes.
\subsubsection{Measurements of the photon-number parity}
 For instance, if one considers the Hilbert space of a single mode of the electromagnetic field 
spanned by the single mode Fock basis $\{\ket n|n=0,1,\ldots,\infty\}$, we can define the photon number parity operator as $\hat P=(-1)^{\hat n}$, where $\hat n$ is the photon number operator fulfilling $\hat n\ket n=n\ket n$. The parity operator has two eigenvalues $\pm 1$ which are both infinitely degenerate and thus can be expressed as in Eq.~(\ref{eq:DirectSumU1}) with $N=1$, $\lambda_1=1$ and $K_1=\infty$. To see this, we write it in the Fock basis \mbox{$\hat P=\sum_{n=0}^\infty \ket{2n}\bra{2n}-\ket{2n+1}\bra{2n+1}$} which is equivalent to $\bigoplus_{j=1}^\infty \hat\sigma_z$ and thus to Eq.~(\ref{eq:DirectSumU1}). According to Eqs.~(\ref{eq:DirectSumU2}) we can define two anti-commuting partners of the parity operator $\hat P=\hat P_z$, which read:
\begin{align}
\hat P_x=\bigoplus_{j=1}^\infty \hat\sigma_x,  \\
\hat P_y=\bigoplus_{j=1}^\infty \hat\sigma_y.
\label{eq:ParityXY}
\end{align}
These kind of parity-pseudospin operators were also used to show that the EPR state can lead to a maximal violation of nonlocality in terms of the CHSH inequality \cite{PhysRevLett.88.040406}. Since they are hermitian and form real Pauli algebra one can consider the present case as an application of the ordinary PMS for qubits (see Table~\ref{square1}) to Hilbert spaces of infinite dimensions. In the following, we will discuss an example where this is not the case since the considered unitary operators do not form a real Pauli algebra.


\subsubsection{Measurements of modular variables}
Our results can also be used to demonstrate state independent contextuality for measurements of observables with continuous spectrum. In particular, we want to formulate a contextuality test that involves measurements of modular variables, as used previously for the demonstration of Bell nonlocality and state-independent contextuality \cite{PhysRevA.82.022114, Asadian:2015aa, PhysRevA.91.012106, PhysRevA.92.062107, PhysRevA.64.062108}. A suitable way of doing so is by using the eigenbasis of the modular position and momentum operators and the formalism developed in~\cite{arXiv:1512.02957} which we will briefly recall here.
%

Every observable in quantum mechanics can be decomposed into a sum of a modular and an integer operator. In the case of a pair of canonically conjugate observables, such as the position and momentum operators, $\hat x$ and $\hat p$, this might be done as follows~\cite{Aharonov}:
\begin{align}
 \hat{x} &= \hat N {\ell}+ \hat{\bar x}, \\
  \hat p &= \hat M \frac{2\pi}{\ell}+ \hat{\bar p},
\label{eq:DefModVar}
\end{align}
where $\hbar=1$, $\hat N$ ($\hat M$) has integer eigenvalues, and $\hat{\bar x}=(\hat x+\ell/4)~\text{mod}[\ell]-\ell/4$ ($\hat{\bar p}=(\hat p+\pi/\ell)~\text{mod} [2\pi/\ell]-\pi/\ell$) is the~{modular position} ({momentum}) operator with eigenvalues in the interval $[-\ell/4,3\ell/4 [$ ($[-\pi/\ell,\pi/\ell[$). Note that the intervals defining the spectrum of the modular operators, $\hat{\bar x}$ and $\hat{\bar p}$, was chosen such that the product of their lengths is equal to $\pi$. This condition assures that the modular parts of $\hat x$ and $\hat p$ commute, $[\hat{\bar x}, \hat{\bar p}]=0$, as shown in Ref.~\cite{Aharonov}. Having in hand a pair of commuting observables we can proceed and define a common set of eigenstates $\ket{\hat{\bar x},\hat{\bar p}}$, which is parametrized by the eigenvalues of the modular operators $\hat{\bar x}$ and $\hat{\bar p}$, and referred to as~{modular basis}. The modular basis $\{\ket{\hat{\bar x},\hat{\bar p}}|\bar x\in [-\ell/4,3\ell/4 [ ,\bar p\in[-\pi/\ell,\pi/\ell[\}$ is complete and thus enables us to represent every state uniquely as:
\begin{equation}
\ket{\Psi}=\int_{-\ell/4}^{3\ell/4}\int_{-\pi/\ell}^{\pi/\ell} d\bar xd\bar p\ \Psi({\bar x},{\bar p})\ \ket{\bar x,\bar p},
\label{eq:StateModVarRep}
\end{equation}
where $\Psi(\bar x,\bar p)$ is a normalized wave function defined on the bounded domain $[-\ell/4,3\ell/4 [\times[-\pi/\ell,\pi/\ell[$. The precise mathematical expressions of the modular eigenstates as superpositions of position or momentum eigenstates, and consequently of the modular wave function $\Psi(\bar x,\bar p)$,  can be found in Ref.~\cite{arXiv:1512.02957}.

Further on, we want to take advantage of the above modular variables formalism to find other observables that allow to demonstrate state-independent contextuality. Note that the modular representation (see Eq.~(\ref{eq:StateModVarRep})) is particularly handy if one deals with states and observables that obey certain periodicities. Let us demonstrate this at the example of the phase-space displacements operator
\begin{align}
\hat D(\nu,\mu)= e^{i \mu \hat x - i\nu \hat p}=e^{-i \mu\nu/2} e^{i \mu \hat x} e^{- i\nu \hat p},
\label{eq:PhaseSpaceDisp}
\end{align}
where $\nu$ and $\mu$ denote displacements in position and momentum, respectively. 
For instance, if we choose $\nu=0$ and $\mu=2\pi/\ell$, we find:
\begin{align}
e^{2\pi i \hat x/\ell}=\int_{-\ell/4}^{\ell/4} d\bar x  \int_{-\pi/\ell}^{\pi/\ell} d\bar p e^{2\pi i \bar x/\ell} \hat\sigma_z(\bar x,\bar p),
\label{eq:SigmaZ} 
\end{align}
where 
\begin{align}\label{eq:pauliModvarZ}
\hat \sigma_z(\bar{x}, \bar{p}) = \ket{\bar{x}, \bar{p}}\bra{\bar{x}, \bar{p}} - \ket{\bar{x} + \ell/2, \bar{p}}\bra{\bar{x}+\ell/2,\bar{p}}. 
\end{align}
Equivalently, by considering the displacements $(\nu=\ell/2,\mu=0)$ and $(\nu=\ell/2,\mu=2\pi/\ell)$, we obtain:
\begin{align}
e^{-i\hat p \ell/2} & =\int_{-\ell/4}^{\ell/4} d\bar x  \int_{-\pi/\ell}^{\pi/\ell} d\bar p e^{-i\bar p \ell/2} \hat\sigma_x(\bar x,\bar p), \label{eq:SigmaX} \\
e^{2\pi \hat x/\ell-i\hat p\ell/2} & = \int_{-\ell/4}^{\ell/4} d\bar x  \int_{-\pi/\ell}^{\pi/\ell} d\bar p e^{i\bar p \ell/2-2\pi i \bar x/\ell} \hat\sigma_y(\bar x,\bar p), \label{eq:SigmaY}
\end{align}
where $\hat \sigma_{x}(\bar{x}, \bar{p})$ and $ \hat \sigma_{y}(\bar{x}, \bar{p})$ are defined through:
\begin{align}
& \hat{\sigma}_x(\bar{x}, \bar{p}) = e^{-i\bar{p}\ell / 2} \ket{\bar{x}, \bar{p}}\bra{\bar{x} + \ell / 2, \bar{p}} + e^{i\bar{p}\ell/2} \ket{\bar{x}, \bar{p}}\bra{\bar{x} + \ell / 2, \bar{p}}, \\
& \hat{\sigma}_y(\bar{x}, \bar{p}) = i(e^{i\bar{p}\ell / 2} \ket{\bar{x}, \bar{p}}\bra{\bar{x} + \ell / 2, \bar{p}} - e^{-i\bar{p}\ell / 2} \ket{\bar{x}, \bar{p}}\bra{\bar{x} + \ell / 2, \bar{p}}).
\end{align}
The operators $\sigma_{x, y, z}(\bar{x}, \bar{p})$ define a Pauli algebra on each of the two dimensional subspaces parametrized by $\bar{x}$ and $\bar{p}$. Consequently, their commutation and anti-commutation relations read:
\begin{align}
[\hat{\sigma}_\alpha(\bar{x}, \bar{p}), \hat{\sigma}_\beta(\bar{x}', \bar{p}')] &= 2i \epsilon_{\alpha\beta\gamma}\hat{\sigma}_\gamma(\bar{x}, \bar{p}) \delta(\bar{x} - \bar{x}') \delta(\bar{p} - \bar{p}'),\label{eq:PauliCommutator}\\
\{\hat{\sigma}_\alpha(\bar{x}, \bar{p}), \hat{\sigma}_\beta(\bar{x}', \bar{p}')\} &= \delta_{\alpha\beta}\mathbb 1(\bar x,\bar p) \delta(\bar{x} - \bar{x}') \delta(\bar{p} - \bar{p}'),
\label{eq:PauliAntiCommutator}
\end{align}
where $\mathbb 1(\bar x,\bar p)=\ket{\bar{x}, \bar{p}}\bra{\bar{x}, \bar{p}} +\ket{\bar{x} + \ell/2, \bar{p}}\bra{\bar{x}+\ell/2,\bar{p}}$, and $\alpha,\beta,\gamma=1,2,3$, representative for $\alpha,\beta,\gamma=x,y,z$, respectively. Using Eqs.~(\ref{eq:PauliCommutator}) and (\ref{eq:PauliAntiCommutator}) it is easy to verify that displacement operators defined by Eqs.~\ref{eq:SigmaZ}, \ref{eq:SigmaX} and \ref{eq:SigmaY} satisfy the relations \eqref{commute2a} and \eqref{commute2b} and thus lead to a maximal state-independent violation of the Peres-Mermin inequality. For example, we have:
\begin{align}
\left[e^{2\pi\hat{x} / \ell}, e^{-i\hat{p}\ell / 2}\right] = & \iint_{-\ell/4}^{\ell/4} d\bar x d\bar x'  \iint_{-\pi/\ell}^{\pi/\ell} d\bar p d\bar p' e^{2\pi i \bar x/\ell}e^{-i\bar p \ell/2} \left[ \hat\sigma_z(\bar x,\bar p), \hat\sigma_x(\bar x', \bar p') \right] \nonumber \\
= & \int_{-\ell/4}^{\ell/4} d\bar x  \int_{-\pi/\ell}^{\pi/\ell} d\bar p e^{2\pi i \bar x/\ell-i\bar p \ell/2} 2i \hat{\sigma}_y(\bar{x}, \bar{p}) \nonumber \\
= &2i e^{-2\pi \hat x/\ell+i\hat p\ell/2},
\end{align}
and:
\begin{align}
\left\lbrace e^{2\pi\hat{x} / \ell}, e^{-i\hat{p}\ell / 2}\right\rbrace = & \iint_{-\ell/4}^{\ell/4} d\bar x   d\bar x' \iint_{-\pi/\ell}^{\pi/\ell} d\bar p d\bar p' e^{2\pi i \bar x/\ell}e^{-i\bar p \ell/2} \{ \hat\sigma_z(\bar x,\bar p), \hat\sigma_x(\bar x', \bar p') \} \nonumber \\
= & 0,
\end{align}
as expected by the relations~\eqref{commute2a} and \eqref{commute2b}. As discussed in Sec.~\ref{sec:PMSarbitrary}, this is a direct consequence of the binary spectral decomposition of the displacement operators (\ref{eq:SigmaZ}), (\ref{eq:SigmaX}) and (\ref{eq:SigmaY}) which have the form of Eqs.~(\ref{eq:DirectSumU1}), (\ref{eq:DirectSumU2}) and (\ref{eq:DirectSumU3}), respectively.

A similar result has been obtained in~\cite{Asadian:2015aa}, where it was shown that for a phase space displacement operator $\mathcal{D}(\alpha_1) = e^{\alpha_1\hat{a}^{\dagger}-\alpha_1^*\hat{a}}$, with $\alpha_1=(\nu_1+i\mu_1)/\sqrt 2$, one can always find two other displacement operators  $\mathcal{D}(\alpha_2)$ and $\mathcal{D}(\alpha_3)$, such that they satisfy the relations~(\ref{pauliRelations1}) and (\ref{pauliRelations2}). The condition for this to hold is that $\alpha_1$, $\alpha_2$ and $\alpha_3$ fulfill the relations $\mathrm{Im}(\alpha_i\alpha_j^*) = \pm\pi/2$ and $\alpha_1 + \alpha_2 + \alpha_3 = 0$. However, it is the modular representation which allows us to write the displacements (\ref{eq:SigmaZ}), (\ref{eq:SigmaX}) and (\ref{eq:SigmaY}), namely those displacements that form a rectangular triangle with area $\pi/2$ in phase space, as a continuous superposition of Pauli operators $\hat\sigma_\beta(\bar x,\bar p)$, with $\beta=x,y,z$ (see  Eq.~(\ref{eq:pauliModvarZ})). Hence, we find that Eqs.~(\ref{eq:SigmaZ}), (\ref{eq:SigmaX}) and (\ref{eq:SigmaY}) are equivalent to the general unitary operators $\hat U_i$, with $i=1,2,3$, defined in Eqs.~(\ref{eq:DirectSumU1}), (\ref{eq:DirectSumU2}) and (\ref{eq:DirectSumU3}), with eigenvalues:
\begin{align}
 \lambda(\bar x,\bar p)&= e^{2\pi i \bar x/\ell}, \\
 \lambda'(\bar x,\bar p)&= e^{ i \bar p \ell/2},  \\
 (\lambda(\bar x,\bar p)\lambda'(\bar x,\bar p))^*&= e^{i\bar p \ell/2-2\pi i \bar x/\ell},
%
\end{align}
for $\hat U_1$, $\hat U_2$ and $\hat U_3$, respectively. Remember that according to our remarks in Sec.~\ref{sec:PMSarbitrary} also $-\lambda(\bar x,\bar p)$, $-\lambda'(\bar x,\bar p)$ and $- (\lambda(\bar x,\bar p)\lambda'(\bar x,\bar p))^*$ are eigenvalues of the three unitary operators, respectively. Hence, we find that $\hat U_1$, $\hat U_2$ and $\hat U_3$ are completely determined by the functions $\lambda(\bar x,\bar p)$ and $\lambda'(\bar x,\bar p)$. In contrast to the case of the parity operator, here all eigenvalues are nondegenerate, \textit{i.e.} $K(\bar x,\bar p)=1$, and we can read the integrals in Eqs.~(\ref{eq:SigmaZ}), (\ref{eq:SigmaX}) and (\ref{eq:SigmaY}) equivalently as a continuous direct sum over Pauli matrices $\hat\sigma_\beta$, with $\beta=x,y,z$, weighted by the functions $\lambda(\bar x,\bar p)$, $\lambda'(\bar x,\bar p)$ and $(\lambda(\bar x,\bar p)\lambda'(\bar x,\bar p))^*$, respectively.

Finally, in order to perform a noncontextuality test we have to measure the real and imaginary parts of the displacement $e^{2\pi \hat x/\ell}$, $e^{-i\hat p\ell/2}$ and $e^{2\pi \hat x/\ell-i\hat p\ell/2}$, according to Eqs.~(\ref{expansion1}) and (\ref{expansion2}), yielding the modular variables $\cos{\left(2\pi \hat x/\ell\right)}$, $\cos{\left(-i\hat p\ell/2\right)}$, $\cos{\left(2\pi \hat x/\ell-i\hat p\ell/2\right)}$, $\sin{\left(2\pi \hat x/\ell\right)}$, $\sin{\left(-i\hat p\ell/2\right)}$ and $\sin{\left(2\pi \hat x/\ell-i\hat p\ell/2\right)}$. Each of these modular variables can be measured indirectly by coupling the considered system to an ancilla qubit and measuring the ancilla state. Possible implementations of such measurements using the transverse degrees of freedom of photons, ions or micro-mechanical oscillators have been proposed in Refs.~\cite{Asadian:2015aa, PhysRevA.91.012106, AsadianMacro}. Sequences of such interferometric measurements thus can be used to measure the correlations contained in Eq.~(\ref{realpartcontextuality}).

\section{Conclusion\label{sec:Conclusion}}
We derived general conditions for an operator to maximally violate non--contextuality inequalities in the Peres-Mermin scenario irrespectively of the dimension of the system used to test it. A consequence of our results is that it is not possible to maximally violate such inequalities for any state using  bipartite systems where one of the systems is in an odd dimensional Hilbert space. Nevertheless, we show how contextuality can be demonstrated using systems of arbitrarily high dimensional subsystems and in continuous variables. In both the discrete and continuous case we find a characterization in terms of their spectrum of observables that can be used to maximally violate the non-contextual bound in the Peres-Mermin inequality. This characterization allow us to find a natural decomposition of the observables in terms of Pauli matrices. Perspectives of our results are implementation of contextuality tests using a wide range of observables both in the discrete and continuous regime and relating the obtained conditions to the possibility of implementing quantum information protocols with continuous variables.

\ack
The authors acknowledge A. Dias Ribeiro and A. Asadian for helpful discussions and CAPES-COFECUB project Ph-855/15  and CNPq for financial support.\\


\end{document}